\documentclass[a4paper,12pt,reqno]{amsart}
\usepackage[utf8]{inputenc}
\usepackage[T1]{fontenc}
\usepackage[british]{babel}
\usepackage{times}

\usepackage[text={16.5cm,25.2cm}, centering, nomarginpar, includeheadfoot, ignoremp]{geometry}
\usepackage{amsmath, amssymb, amsfonts, dsfont, amsthm, mathrsfs, bm, mathtools, eqnarray, faktor}
\usepackage{scalerel, stackengine}
\usepackage{hyperref}
\usepackage{xparse}
\usepackage{listings,enumitem}
\usepackage[hang,small,bf]{caption}
\usepackage{fancyhdr}
\usepackage{indentfirst}

\usepackage[toc,page]{appendix}
\theoremstyle{definition}

\theoremstyle{plain}

\newtheorem{theo}{Theorem}[section]
\newtheorem{lemma}[theo]{Lemma}
\newtheorem{prop}[theo]{Proposition}

\numberwithin{equation}{section}
\numberwithin{defn}{section}
\numberwithin{note}{section}

\newcommand{\vs}{\vspace{0.5cm}}
\newcommand{\+}{\mspace{1.5mu}}

\stackMath
\renewcommand\widehat[1]{%
\savestack{\tmpbox}{\stretchto{%
    \scaleto{%
        \scalerel*[\widthof{\ensuremath{#1}}]{\kern.1pt\mathchar"0362\kern.1pt}%
        {\rule{0ex}{\textheight}}
    }{\textheight}%
}{2.4ex}}%
\stackon[-6.9pt]{#1}{\tmpbox}%
}
\renewcommand\vec[1]{%
    \boldsymbol{#1}
}

\allowdisplaybreaks    

\title[Three Bosons with Contact Interactions]{Some Remarks on the Regularized Hamiltonian for Three Bosons with Contact Interactions}

\author[D. Ferretti]{Daniele Ferretti}
\address[D. Ferretti]{Department of Mathematics G. Castelnuovo, University of Rome ``La Sapienza''\\
Piazzale  Aldo Moro, 5, 00185 Rome, Italy}
\email{d.ferretti@uniroma1.it}

\author[A. Teta]{Alessandro Teta}
\address[A. Teta]{Department of Mathematics G. Castelnuovo, University of Rome ``La Sapienza'',\newline
Piazzale  Aldo Moro, 5, 00185 Rome, Italy}
\email{teta@mat.uniroma1.it}
\date{}

\thanks{The authors acknowledge the support of the GNFM Gruppo Nazionale per la Fisica Matematica - INdAM}

\begin{document}

\pagenumbering{Alph}
\makeatletter
\renewcommand\subsection{%
  \@startsection{subsection}%
    {2}
    {0em}
    {-1ex \@plus 0.1ex \@minus -0.05ex}
    {-1em \@plus 0.2em}
    {\scshape}
  }

\makeatother

\begin{abstract}
We discuss some properties of a model Hamiltonian for a system of three bosons interacting via zero-range forces in three dimensions.
In order to avoid the well known instability phenomenon, we consider a regularized Hamiltonian with a repulsive three-body force. 
We review the main result recently obtained in~\cite{BCFT}  where, starting from a suitable quadratic form $Q$, the self-adjoint and bounded from below Hamiltonian $\mathcal H$ is constructed provided that the strength $\gamma$ of the three-body repulsion is larger than a threshold parameter $\gamma_c\+$. 
We also show that the threshold value $\gamma_c$ found in~\cite{BCFT} is optimal, in the sense that the quadratic form $Q$ is unbounded from below if $\gamma<\gamma_c\+$. Finally, we give an alternative and much simpler proof of the result in \cite{BCFT} whenever $\gamma > \gamma'_c\+$, with $\gamma'_c$ strictly larger than $\gamma_c\+$.
\end{abstract}
\maketitle
\begin{footnotesize}
\emph{Keywords: Zero-range interactions; Three-body Hamiltonians; Schr\"odinger operators.} 
 
\emph{MSC 2020: 
    81Q10; 
    81Q15; 
    70F07; 
    46N50; 
}  
\end{footnotesize}

\pagenumbering{arabic}
\vs\vs
\section{Introduction}
\vs
In this note we discuss some properties of a model Hamiltonian describing the dynamics of three identical bosons interacting via zero-range forces in dimension three.
Since the seminal papers by Minlos and Faddeev (\cite{MF, MF2}), it is known that a natural candidate for such Hamiltonian turns out to be  unbounded from below, giving rise to the so-called Thomas effect.
Here natural means that the boundary condition defining the Hamiltonian (known as TMS boundary condition) is the direct generalization to the three-body case of the boundary condition characterizing the Hamiltonian of the two-body problem.
Roughly speaking, the reason of such instability is due to the interaction becoming too singular as all the three particles are close to each other.
We note that this pathology is absent in dimension one, where perturbation theory of quadratic forms can be used, and in dimension two, where the renormalized two-body boundary condition is sufficient to avoid the collapse (see e.g.~\cite{DFT, DR}). 

It is worth to underline that the construction of a self-adjoint and bounded from below Hamiltonian for three, or more, interacting bosons with zero-range forces in dimension three is a challenging open problem in Mathematical Physics.
Following a suggestion contained in~\cite{MF}, it has been recently studied (\cite{FiT}, \cite{Miche}, \cite{BCFT}) a regularized version of the Hamiltonian for a system of three bosons (see also~\cite{FeT} for the case of $N$ bosons interacting with an impurity).
The main idea is to introduce a three-body repulsion that reduces to zero the strength of the contact interaction between two particles if the third particle approaches the common position of the first two.
On the other hand, when the third particle is far enough, the usual two-body point interaction is restored.
The result is that the regularized Hamiltonian is self-adjoint and bounded from below if the strength $\gamma$ of the three-body interaction is larger than a threshold value $\gamma_c\+$. 

The aim of this paper is to describe the construction of such regularized Hamiltonian following the approach developed in~\cite{BCFT} and also to prove two further results.
More precisely, in section~\ref{regham} we introduce the notation and we formulate the main result of~\cite{BCFT}, essentially based on the analysis of a suitable quadratic form $Q$.

\noindent In section~\ref{optgam} we prove that the threshold value $\gamma_c$  obtained in~\cite{BCFT} is optimal, in the sense that for $\gamma < \gamma_c$ the quadratic form $Q$ is unbounded from below.

\noindent In section~\ref{anapos} we give a different proof of the main result in~\cite{BCFT} based on a new  approach in position space.
The proof is surely less general since it is valid only for $\gamma> \gamma'_c\+$, where $\gamma'_c > \gamma_c\+$.
On the other hand it has the advantage to be much simpler and to show that the choice of the three-body force is not arbitrary but it is dictated by the inherent singularity of the problem.

\vs\vs
\section{Regularized Hamiltonian}\label{regham}
\vs

Let us consider a system composed of three identical spinless bosons of mass $\frac{1}{2}$ in three dimensions and let us fix the center of mass reference frame so that $\vec{x}_1$, $\vec{x}_2$ and $\vec{x}_3=-\vec{x}_1-\vec{x}_2$ represent the Cartesian coordinates of the three particles.
We also introduce the Jacobi coordinates
\begin{equation}
    \begin{cases}
        \vec{r}_k=\frac{1}{2}\sum_{i,j=1}^3\epsilon_{ijk}(\vec{x}_i-\vec{x}_j),\\
        \vec{\rho}_k=\frac{3}{2}\,\vec{x}_k-\frac{1}{2}\sum_{\ell=1}^3 \vec{x}_\ell,
    \end{cases}\qquad k\in\{1,2,3\}
\end{equation}
where $\epsilon_{ijk}$ is the Levi-Civita symbol,
so that one has the following identities
\begin{equation}
    \begin{cases}
        \vec{r}_{k\pm 1}=-\frac{1}{2}\vec{r}_k\mp \vec{\rho}_k,\\
        \vec{\rho}_{k\pm 1}=\pm\frac{3}{4}\vec{r}_k-\frac{1}{2} \vec{\rho}_k,
    \end{cases}\qquad k\in\mathbb{Z}\mspace{1.5mu}\diagup_{\mspace{-4.5mu}\{3\}}.
\end{equation}
Denoting by $\vec{x}=\vec{r}_1=\vec{x}_2-\vec{x}_3$ and $\vec{y}=\vec{\rho}_1=\vec{x}_1-\frac{\vec{x}_2\,+\,\vec{x}_3}{2}$, the Hilbert space of the system is
\begin{equation}\label{hilbertSpace}
    L^2_{\mathrm{sym}}(\mathbb{R}^6):=\left\{\psi\in L^2(\mathbb{R}^6)\,\big|\;\psi(\vec{x},\vec{y})=\psi(-\vec{x},\vec{y})=\psi\!\left(\tfrac{1}{2}\,\vec{x}+\vec{y},\tfrac{3}{4}\,\vec{x}-\tfrac{1}{2}\,\vec{y}\right)\right\}\!.
\end{equation}
Indeed, notice that the symmetry conditions in~\eqref{hilbertSpace} corresponds to the exchange of particles $2,3$ and $1,2$ that implies also the condition $\psi(\vec{x},\vec{y})=\psi\!\left(\tfrac{1}{2}\,\vec{x}-\vec{y},-\tfrac{3}{4}\,\vec{x}-\tfrac{1}{2}\,\vec{y}\right)$, associated with the exchange of particles $3,1$.
If the bosons interact via zero-range forces, then the system is described, at least formally, by the Hamiltonian
\begin{equation}\label{formalH}
    \hat{\mathcal{H}}=\mathcal{H}_0+\mu\mspace{1.5mu}\delta(\vec{x})+\mu\mspace{1.5mu}\delta(\tfrac{1}{2}\,\vec{x}+\vec{y})+\mu\mspace{1.5mu}\delta(\tfrac{1}{2}\,\vec{x}-\vec{y})
\end{equation}
where $\mu\!\in\!\mathbb{R}$ is a coupling constant and $\mathcal{H}_0$ is the free Hamiltonian of the system, i.e.
\begin{equation}
    \mathcal{H}_0=-\Delta_{\vec{x}}-\tfrac{3}{4}\,\Delta_{\vec{y}} .
\end{equation}
In order to define a rigorous counterpart of $\hat{\mathcal{H}}$, one needs to build a perturbation of the free Hamiltonian supported on the coincidence hyperplanes
\begin{subequations}
\begin{equation}
    \pi_k:=\left\{(\vec{r}_k,\vec{\rho}_k)\!\in\mathbb{R}^6\,\big|\;\vec{r}_k\!=\vec{0}\right\},\qquad \pi:=\textstyle{\bigcup_{k=1}^3 \pi_k}
\end{equation}
or, equivalently,
\begin{equation}
    \begin{split}
        \pi_1:=\!\left\{(\vec{x},\vec{y})\!\in\mspace{-1.5mu}\mathbb{R}^6\right.&\left.\!\big|\;\vec{x}=\vec{0}\right\}\!,\quad \pi_2:=\!\left\{(\vec{x},\vec{y})\!\in\mspace{-1.5mu}\mathbb{R}^6\,\big|\;\vec{y}=-\tfrac{1}{2}\,\vec{x}\right\}\!,\\
        &\pi_3:=\!\left\{(\vec{x},\vec{y})\!\in\mspace{-1.5mu}\mathbb{R}^6\,\big|\;\vec{y}=\tfrac{1}{2}\,\vec{x}\right\}\!.
    \end{split}
\end{equation}
\end{subequations}
In other words, we look for a self-adjoint and bounded from below extension in $L^2_{\mathrm{sym}}(\mathbb{R}^6)$ of the following symmetric and densely defined operator
\begin{equation}
    \dot{\mathcal{H}}_0:=\mathcal{H}_0\big\rvert_{\mathscr{D}(\dot{\mathcal{H}}_0)},\qquad\mathscr{D}(\dot{\mathcal{H}}_0):=H^2_0(\mathbb{R}^6\setminus\pi)\cap L^2_{\mathrm{sym}}(\mathbb{R}^6)
\end{equation}
that is closed according to the graph norm of $\mathcal{H}_0, \mathscr{D}(\mathcal{H}_0)=H^2(\mathbb{R}^6)\cap L^2_{\mathrm{sym}}(\mathbb{R}^6)$.
In particular, we are interested in the family of self-adjoint extensions studied in~\cite{BCFT} (see also~\cite{Miche}) which, at least formally, are characterized by the boundary condition 
\begin{equation}\label{mfBC}
 \psi(\vec{x},\vec{y})=\frac{\xi(\vec{y})}{\lvert\vec{x}\rvert}+\alpha(\vec{y})\mspace{1.5mu}\xi(\vec{y})+o(1),\qquad \lvert\vec{x}\rvert\rightarrow 0,
\end{equation}
where $\alpha$ is a position dependent parameter given by
\begin{equation}
    \alpha:\vec{y}\longmapsto -\frac{1}{\mathfrak{a}}+\frac{\gamma}{\lvert\vec{y}\rvert} \,\theta(\lvert\vec{y}\rvert)
\end{equation}
with $\mathfrak{a}$ the two-body scattering length, $\gamma$ a positive  parameter representing the strength of the regularization and $\theta$ a real measurable function with compact support such that
\begin{equation}\label{technicalThetaAssumption}
    1-\frac{s}{b}\leq\theta(s)\leq 1+\frac{s}{b}\,, \qquad s\geq 0
\end{equation}
for some $b\!>\mspace{-1.5mu}0$.
Notice that assumption~\eqref{technicalThetaAssumption} forces the function $\theta$ to be continuous at zero, with $\theta(0)=1$.
Moreover, the simplest choice for $\theta$ is the characteristic function of the ball of radius $b$ centered in the origin. We also stress that, due to the symmetry constraints of $L^2_{\mathrm{sym}}(\mathbb{R}^6)$, the boundary condition~\eqref{mfBC} implies
\begin{gather*}
    \psi(\vec{x},\vec{y})=\frac{\xi(\vec{x})}{\left\lvert\vec{y}+\frac{1}{2}\,\vec{x}\right\rvert}+\alpha(\vec{x})\mspace{1.5mu}\xi(\vec{x})+o(1),\;\;\;\qquad \vec{y}\rightarrow -\tfrac{1}{2}\mspace{1.5mu}\vec{x},\\
    \psi(\vec{x},\vec{y})=\frac{\xi(-\vec{x})}{\left\lvert\vec{y}-\frac{1}{2}\,\vec{x}\right\rvert}+\alpha(-\vec{x})\mspace{1.5mu}\xi(-\vec{x})+o(1),\qquad \vec{y}\rightarrow \tfrac{1}{2}\mspace{1.5mu}\vec{x}.
\end{gather*}
Observe that for $\gamma=0$ equation~\eqref{mfBC} reduces to the standard TMS boundary condition, which leads to the Thomas effect. Then, for $\gamma>0$ we are introducing a three-body repulsion meant to regularize the ultraviolet singularity occurring when the positions of all particles coincide.
However, since $\mathrm{supp}\,\theta$ is a compact, the usual two-body point interaction is restored when the third particle is far enough.

\vs
\noindent The procedure adopted in~\cite{BCFT} for the rigorous construction of the Hamiltonian is the following: one first introduces the quadratic form $Q$ in $L^{2}_{\mathrm{sym}}(\mathbb{R}^6)$ describing, at least formally, the expectation value of the energy of our three-body  system.
Then one defines a suitable form domain $\mathscr{D}(Q)$ and proves that $Q, \mathscr{D}(Q)$ is closed and bounded from below.
Finally, the Hamiltonian is defined as the unique self-adjoint and bounded from below operator associated to the quadratic form. 

\noindent In order to define the quadratic form $Q, \mathscr{D}(Q)$, we first introduce an auxiliary  hermitian quadratic form $\Phi^\lambda$ in $L^2(\mathbb{R}^3)$ given by~\cite[equation (3.1)]{BCFT}, namely
\begin{equation}\label{phiDef}
    \Phi^\lambda:=\Phi^\lambda_{\mathrm{diag}}\!+\Phi^\lambda_{\mathrm{off}}+\Phi_{\mathrm{reg}}\mspace{-1.5mu}+\Phi_0,\qquad \mathscr{D}(\Phi^\lambda)=H^{1/2}(\mathbb{R}^3),
\end{equation}
where
\begin{subequations}\label{phiComponents}
\begin{align}
    \Phi^\lambda_{\mathrm{diag}}[\xi]&:=12\pi\!\int_{\mathbb{R}^3}\mspace{-7.5mu}d\vec{p}\,\sqrt{\tfrac{3}{4}\mspace{1.5mu}p^2+\lambda\mspace{1.5mu}}\,\lvert\hat{\xi}(\vec{p})\rvert^2,\label{diagPhi}\\
    \Phi^\lambda_{\mathrm{off}}[\xi]&:=-\frac{12}{\pi}\!\int_{\mathbb{R}^6}\mspace{-7.5mu}d\vec{p}\+d\vec{q}\;\frac{\overline{\hat{\xi}(\vec{p})}\,\hat{\xi}(\vec{q})}{p^2+q^2+\vec{p}\!\cdot\mspace{-1.5mu}\vec{q}+\lambda},\label{offPhi}\\
    \Phi_{\mathrm{reg}}[\xi]&:=\frac{6\gamma}{\pi}\!\int_{\mathbb{R}^6}\mspace{-7.5mu}d\vec{p}\+d\vec{q}\;\frac{\overline{\hat{\xi}(\vec{p})}\,\hat{\xi}(\vec{q})}{|\vec{p}-\vec{q}|^2},\label{regPhi}\\
    \Phi_0[\xi]&:=12\pi\!\int_{\mathbb{R}^3}\mspace{-7.5mu}d\vec{y}\:\beta(\vec{y})\lvert\xi(\vec{y})\rvert^2,\qquad\beta:\vec{y}\longmapsto -\frac{1}{\mathfrak{a}}+\gamma\,\frac{\theta(y)-1}{y}.\label{notePhi}
\end{align}
\end{subequations}
By assumption~\eqref{technicalThetaAssumption}, one has $\beta\mspace{-1.5mu}\in\mspace{-1.5mu} L^\infty(\mathbb{R}^3)$ and therefore $\Phi_0$ is bounded.
The proof of the fact that $\Phi^\lambda$ is well defined in $H^{1/2}(\mathbb{R}^3)$ is relatively standard and it is given in~\cite[proposition 3.1]{BCFT}.
The more relevant point concerning $\Phi^\lambda$ is that it is coercive for $\lambda$ large enough as long as $\gamma>\gamma_c\+$, with
\begin{equation}\label{criticalGamma}
    \gamma_c=\frac{4}{3}-\frac{\sqrt{3}}{\pi}\approx 0.782004.
\end{equation}
The proof is given in~\cite[proposition 3.6]{BCFT} and it requires a rather long and non trivial analysis performed in the momentum representation.
The conclusion is that there exists $\lambda_0\mspace{-1.5mu}>\mspace{-1.5mu}0$ such that $\Phi^\lambda$ is closed and bounded from below by a positive constant for each $\lambda\!>\!\lambda_0$ and $\gamma>\gamma_c\+$.
Therefore one can uniquely define a self-adjoint, positive and invertible operator $\Gamma^\lambda$ in $L^2(\mathbb{R}^3)$ such that
\begin{equation}
    \Phi^\lambda[\xi]=\langle\xi,\,\Gamma^\lambda\xi\rangle_{L^2(\mathbb{R}^3)},\qquad\forall\,\xi\in D
\end{equation}
with $D=\mathscr{D}(\Gamma^\lambda)$ a dense subspace independent of $\lambda$.
Furthermore, defining the continuous\footnote{Here $\mathscr{D}(\mathcal{H}_0)$ must be intended as a Hilbert subspace of $L^2_{\mathrm{sym}}(\mathbb{R}^6)$ endowed with the graph norm of $\mathcal{H}_0$.} operator
\begin{equation}\label{traceOperator}
    \begin{split}
        \tau:\mathscr{D}(\mathcal{H}_0)\longrightarrow L^2(\mathbb{R}^3),\\
        \varphi\longmapsto 12\pi \,\varphi\rvert_{\pi_1}
    \end{split}
\end{equation}
satisfying $\mathrm{ran}(\tau) = H^{1/2}(\mathbb{R}^3)$ and $\mathrm{ker}(\tau)=\mathscr{D}(\dot{\mathcal{H}}_0)$, one can check that the injective operator $G(z):=(\tau R_{\mathcal{H}_0}(\bar{z}))^\ast\!\in\mathscr{B}(L^2(\mathbb{R}^3),L^2_{\mathrm{sym}}(\mathbb{R}^6))$ is represented in the Fourier space by
\begin{equation}\label{potential}
    (\widehat{G(z)\mspace{1.5mu}\xi})(\vec{k},\vec{p})=\sqrt{\frac{2}{\pi}}\;\frac{\hat{\xi}(\vec{p})+\hat{\xi}\!\left(\vec{k}\!-\frac{1}{2}\vec{p}\right)\mspace{-1.5mu}+\hat{\xi}\!\left(-\vec{k}\!-\frac{1}{2}\vec{p}\right)\!}{k^2+\frac{3}{4}p^2-z},\qquad z\in\rho(\mathcal{H}_0).
\end{equation}
We are now in position to introduce the quadratic form in $L^2_{\mathrm{sym}}(\mathbb{R}^6)$ (\cite[definition 2.1]{BCFT})
\begin{align}
    \mathscr{D}(Q)&:=\!\left\{\psi\mspace{-1.5mu}\in\mspace{-1.5mu} L^2_{\mathrm{sym}}(\mathbb{R}^6)\,\big|\; \psi\mspace{-1.5mu}=\phi_\lambda\mspace{-1.5mu}+G(-\lambda)\mspace{1.5mu}\xi,\;\phi_\lambda\!\in\mspace{-1.5mu} H^1(\mathbb{R}^6),\; \xi\mspace{-1.5mu} \in\mspace{-1.5mu} H^{1/2}(\mathbb{R}^3),\;\lambda\!>\mspace{-1.5mu}0\right\}\!,\nonumber\\
    Q[\psi]&:=\lVert\mathcal{H}_0^{1/2}\phi_\lambda\rVert^2+\lambda\lVert\phi_\lambda\rVert^2-\lambda\lVert\psi\rVert^2+\Phi^\lambda[\xi].\label{quadraticForm}
\end{align}
Using the properties of  $\Phi^\lambda$ and $G(-\lambda)$, it is now easy to show that the above quadratic form is closed and bounded from below if $\gamma > \gamma_c\+$.
Hence it uniquely defines  a self-adjoint and lower semi-bounded operator $\mathcal H$ which, by definition, is the Hamiltonian of the three bosons system.

\vs
\noindent Following an equivalent approach, one can consider the densely defined and closed operator $\Gamma(z)$, given by
\begin{equation}
    \begin{split}
        &\Gamma(z):D\subset L^2(\mathbb{R}^3)\longrightarrow L^2(\mathbb{R}^3)\\
        \Gamma(z):=\Gamma^\lambda-&(\lambda+z)\mspace{1.5mu}G(\bar{z})^\ast G(-\lambda),\qquad \lambda\!>\!\lambda_0,\,z\in\rho(\mathcal{H}_0)
    \end{split}
\end{equation}
which represents a sort of analytic continuation of $\Gamma^\lambda$, $D$.
Actually, one can prove that $\Gamma(z)$ fulfils
\begin{subequations}
\begin{gather}
    \Gamma(z)^\ast=\Gamma(\bar{z}),\qquad\forall\,z\in\rho(\mathcal{H}_0),\\
    \Gamma(z)-\Gamma(w)=(z-w)\mspace{1.5mu}G(\bar{z})^\ast G(w),\qquad\forall\,w,z\in\rho(\mathcal{H}_0),\\
    \forall\,z \in\mathbb{C}: \;\Re(z)<-\lambda_0\:\lor\: \Im(z)>0,\quad 0\in\rho(\Gamma(z)).
\end{gather}
\end{subequations}
These properties imply, according to e.g.~\cite{P}, that for any $z\in\mathbb{C}$ such that $\Gamma(z)$ has a bounded inverse, the operator
\begin{equation}\label{resolvent}
    R(z)=(\mathcal{H}_0 -z)^{-1}+G(z)\Gamma(z)^{-1}G(\bar{z})^\ast\,
\end{equation}
defines the resolvent of a self-adjoint and bounded from below operator which coincides with the Hamiltonian $\mathcal{H}$ obtained with the approach based on the quadratic form and $$\left\{z\in\mathbb{C}\,\big|\;\Re(z)<-\lambda_0\:\lor\: \Im(z)>0\right\}\!\subseteq\mspace{-1.5mu}\rho(\mathcal{H}).$$
Moreover, one can verify that $\mathcal{H}$ coincides with $\mathcal{H}_0$ on $\mathscr{D}(\dot{\mathcal{H}}_0)$, satisfies the boundary condition~\eqref{mfBC} in the $L^2$ sense (see~\cite[remark 4.1]{BCFT}) and it is characterized by
\begin{equation}\label{hamiltonian}
    \begin{split}
    \mathscr{D}(\mathcal{H})&=\!\left\{\psi\mspace{-1.5mu}\in \mathscr{D}(Q)\,\big|\; \phi_z\!\in\mspace{-1.5mu} \mathscr{D}(\mathcal{H}_0),\; \xi\mspace{-1.5mu} \in\mspace{-1.5mu} D,\,\Gamma(z)\mspace{1.5mu}\xi = \tau\phi_z\right\}\!,\\
    \mathcal{H}\psi&=\mathcal{H}_0\mspace{1.5mu}\phi_z+z\mspace{1.5mu} G(z)\mspace{1.5mu}\xi.
\end{split}
\end{equation}

\vs\vs
\section{Optimality of \texorpdfstring{$\gamma_c$}{TEXT}}\label{optgam}
\vs

In this section we prove the optimality of the threshold parameter $\gamma_c$ defined by~\eqref{criticalGamma}.
More precisely our goal is to prove the following theorem.
\begin{theo}\label{unboundedHamiltonian}
    Whenever $\gamma<\gamma_c$, the quadratic form $Q$ given by~\eqref{quadraticForm} is unbounded from below.
\end{theo}
\noindent In order to achieve the result, we shall adapt the ideas contained in~\cite[section 5]{FT}.
Denote for short $G^\lambda\!:=\mspace{-1.5mu}G(-\lambda)$ for any $\lambda\!>\!0$ and let $\{u_n\}_{n\in\,\mathbb{N}}\!\subset\!\mathscr{D}(Q)$ be a sequence of trial functions given by
\begin{gather}
    u_n(\vec{x},\vec{y})=(G^\lambda\mspace{1.5mu}\eta_n)(\vec{x},\vec{y}),\\
    \eta_n(\vec{y})=n^2\mspace{1.5mu}f(n\mspace{1.5mu}\vec{y}),\qquad f\in H^{1/2}(\mathbb{R}^3).\label{trialCharge}
\end{gather}
We stress that, by an explicit estimate due to~\eqref{potential}, one finds
\begin{equation}
    \inf_{n\in\mspace{1.5mu}\mathbb{N}}\,\lVert G^\lambda \eta_n\rVert_{L^2_{\mathrm{sym}}(\mathrm{R}^6)} \gneq 0.
\end{equation}
Indeed,
\begin{align*}
    \lVert G^\lambda\eta_n\rVert^2&=\frac{2}{\pi}\!\int_{\mathbb{R}^6}\mspace{-7.5mu}d\vec{k}d\vec{p}\:\frac{1}{n^2}\frac{\left\lvert\hat{f}\!\left(\tfrac{\vec{p}}{n}\right)+\hat{f}\mspace{-1.5mu}\left(\tfrac{\vec{k}}{n}-\tfrac{\vec{p}}{2n}\right)+\hat{f}\mspace{-1.5mu}\left(\tfrac{\vec{k}}{n}+\tfrac{\vec{p}}{2n}\right)\right\rvert^2}{\left(k^2+\frac{3}{4}p^2+\lambda\right)^2}\\
    &=\frac{2}{\pi}\!\int_{\mathbb{R}^6}\mspace{-7.5mu}d\vec{\kappa}d\mspace{-1.5mu}\vec{q}\;\frac{\left\lvert\hat{f}(\vec{q})+\hat{f}\mspace{-1.5mu}\left(\vec{\kappa}-\tfrac{\vec{q}}{2}\right)+\hat{f}\mspace{-1.5mu}\left(\vec{\kappa}+\tfrac{\vec{q}}{2}\right)\right\rvert^2}{\left(\kappa^2+\frac{3}{4}q^2+\frac{\lambda}{n^2}\right)^2}>\lVert G^\lambda f\rVert^2.
\end{align*}
Our goal is to show that whenever $\gamma$ is smaller than the threshold value $\gamma_c$ given by~\eqref{criticalGamma}, one has
\begin{equation}
    \lim_{n\rightarrow\, +\infty} Q[u_n] =-\infty,
\end{equation}
According to~\eqref{quadraticForm}, we have
\begin{equation}
    Q[u_n]=-\lambda\,\lVert G^\lambda\mspace{1.5mu}\eta_n\rVert_{L^2_{\mathrm{sym}}(\mathrm{R}^6)}+\Phi^\lambda[\eta_n]\leq-\lambda\,\lVert G^\lambda f\rVert_{L^2_{\mathrm{sym}}(\mathrm{R}^6)}+\Phi^\lambda[\eta_n]
\end{equation}
and therefore the theorem is proven if we exhibit some $f\!\in\mspace{-1.5mu}H^{1/2}(\mathbb{R}^3)$ such that
\begin{equation}\label{phiUnboundedBelow}
    \lim_{n\rightarrow\,+\infty}\Phi^\lambda[\eta_n]=-\infty.
\end{equation}
\begin{lemma}\label{leadingOrderLemma}
    Let $\Phi^\lambda$ and $\eta_n$ be given by~\eqref{phiDef} and~\eqref{trialCharge}, respectively.
    Then, one has
    \begin{equation*}
        \Phi^\lambda[\eta_n]=n^2(\Phi^0_{\mathrm{diag}}\!+\Phi^0_{\mathrm{off}}\mspace{-1.5mu}+\Phi_{\mathrm{reg}})[f]+\mathcal{O}(n).
    \end{equation*}
    \begin{proof}
        First of all, we can neglect the bounded component $\Phi_0\+$, since
        \begin{align*}
            \Phi_0[\eta_n]&=12\pi n^4\!\!\int_{\mathbb{R}^3}\mspace{-7.5mu}d\vec{y}\:\beta(y)\,\lvert f(n\mspace{1.5mu}\vec{y})\rvert^2=12\pi\mspace{1.5mu} n\!\int_{\mathbb{R}^3}\mspace{-7.5mu}d\vec{t}\:\beta\!\left(\tfrac{t}{n}\right)\lvert f(\vec{t})\rvert^2\\
            &\leq 12\pi\mspace{1.5mu} n\,\lVert \beta\rVert_{L^\infty(\mathbb{R}^3)}\mspace{1.5mu}\lVert f\rVert^2_{L^2(\mathbb{R}^3)}\implies \Phi_0[\eta_n]=\mathcal{O}(n),\qquad n\rightarrow +\infty.
        \end{align*}
        Next, rescaling properly the variables in computing $\Phi^\lambda_{\mathrm{diag}}[\eta_n]$, one gets
        \begin{align*}
            \Phi^\lambda_{\mathrm{diag}}&[\eta_n]=12\pi\mspace{1.5mu}n\! \int_{\mathbb{R}^3}\mspace{-7.5mu}d\vec{\kappa}\: \sqrt{\tfrac{3}{4}\, n^2\kappa^2 +\lambda\,}\,\lvert\hat{f}(\vec{\kappa})\rvert^2\\
            &=6\sqrt{3}\,\pi\mspace{1.5mu}n^2\! \int_{\mathbb{R}^3}\mspace{-7.5mu}d\vec{\kappa}\: \lvert\vec{\kappa}\rvert\,\lvert\hat{f}(\vec{\kappa})\rvert^2+12\pi\mspace{1.5mu}n\!\int_{\mathbb{R}^3}\mspace{-7.5mu}d\vec{\kappa}\mspace{1.5mu} \left(\mspace{-1.5mu}\sqrt{\tfrac{3}{4}\,n^2\kappa^2+\lambda\,}-\sqrt{\tfrac{3}{4}}\,n\kappa\right)\mspace{-1.5mu}\lvert\hat{f}(\vec{\kappa})\rvert^2\\
            &=n^2\Phi^{\mspace{1.5mu}0}_{\mathrm{diag}}[f]+o(n).
        \end{align*}
        Indeed, exploiting the elementary inequality
        $\sqrt{a^2+b^2}-\lvert a \rvert\leq \lvert b\rvert$, we can use the dominated convergence theorem to obtain
        \begin{equation*}
            \lim_{n\rightarrow +\infty}\int_{\mathbb{R}^3}\mspace{-7.5mu}d\vec{\kappa} \left(\mspace{-1.5mu}\sqrt{\tfrac{3}{4}\,n^2\kappa^2+\lambda\,}-\sqrt{\tfrac{3}{4}}\,n\kappa\right)\mspace{-1.5mu}\lvert\hat{f}(\vec{\kappa})\rvert^2=0.
        \end{equation*}
        Concerning the regularizing contribution, one simply has
        \begin{align*}
            \Phi_{\mathrm{reg}}[\eta_n]&=\frac{6\gamma}{\pi}\mspace{1.5mu}n^2\!\int_{\mathbb{R}^3}\mspace{-7.5mu}d\vec{p}\! \int_{\mathbb{R}^3}\mspace{-7.5mu} d\vec{q}\;\frac{\overline{\hat{f}(\vec{p})}\,\hat{f}(\vec{q})}{\lvert\vec{p}-\vec{q}\rvert^2}\\
            &=n^2 \Phi_{\mathrm{reg}}[f].
        \end{align*}
      Finally, we compute $\Phi^\lambda_{\mathrm{off}}[\eta_n] $
        \begin{align*}
            \Phi^\lambda_{\mathrm{off}}&[\eta_n] =-\frac{12}{\pi}\mspace{1.5mu}n^2\!\int_{\mathbb{R}^6}\mspace{-7.5mu}d\vec{p}\+d\vec{q}\;\frac{\overline{\hat{f}(\vec{p})}\,\hat{f}(\vec{q})}{p^2+q^2+\vec{p}\!\cdot\mspace{-1.5mu}\vec{q}+\frac{\lambda}{n^2}}\\
            = &-\frac{12}{\pi}\mspace{1.5mu}n^2\!\int_{\mathbb{R}^6}\mspace{-7.5mu}d\vec{p}\+d\vec{q}\;\frac{\overline{\hat{f}(\vec{p})}\,\hat{f}(\vec{q})}{p^2+q^2+\vec{p}\!\cdot\mspace{-1.5mu}\vec{q}}\,+\\
            &+\frac{12}{\pi}\mspace{1.5mu}\lambda\!\int_{\mathbb{R}^6}\mspace{-7.5mu}d\vec{p}\+d\vec{q}\;\frac{\overline{\hat{f}(\vec{p})}\,\hat{f}(\vec{q})}{\big(p^2+q^2+\vec{p}\!\cdot\mspace{-1.5mu}\vec{q}+\frac{\lambda}{n^2}\big)\big(p^2+q^2+\vec{p}\!\cdot\mspace{-1.5mu}\vec{q}\big)}.
        \end{align*}
        Defining the integral operator in $L^2(\mathbb{R}^3)$ given by
        \begin{equation}
            (P_{\!n}\+\hat{\varphi})(\vec{p}):=\frac{12}{\pi}\mspace{1.5mu}\lambda\!\int_{\mathrm{R}^3}\mspace{-7.5mu}d\vec{q}\;\frac{\hat{\varphi}(\vec{q})}{\big(p^2+q^2+\vec{p}\!\cdot\mspace{-1.5mu}\vec{q}+\frac{\lambda}{n^2}\big)\big(p^2+q^2+\vec{p}\!\cdot\mspace{-1.5mu}\vec{q}\big)},
        \end{equation}
        we can write
        \begin{equation*}
            \Phi^\lambda_{\mathrm{off}}[\eta_n]=n^2\Phi^{\mspace{1.5mu}0}_{\mathrm{off}}[f]+\!\int_{\mathrm{R}^3}\mspace{-7.5mu}d\vec{p}\:\overline{\hat{f}(\vec{p})}\,(P_{\!n}\+\hat{f})(\vec{p}).
        \end{equation*}
        We notice that  $P_{\!n}$ is a Hilbert-Schmidt operator and 
        \begin{align*}
            \lVert P_{\!n}\rVert^2_{\mathscr{B}(L^2(\mathbb{R}^3))}\!&\leq\frac{124\mspace{1.5mu}\lambda^2\!}{\pi^2}\! \int_{\mathbb{R}^6}\mspace{-7.5mu}d\vec{p}\+d\vec{q}\;\frac{1}{\big(p^2+q^2+\vec{p}\!\cdot\mspace{-1.5mu}\vec{q}+\frac{\lambda}{n^2}\big)^2\big(p^2+q^2+\vec{p}\!\cdot\mspace{-1.5mu}\vec{q}\big)^2}\\
            &\mspace{-18mu}\leq\frac{124\mspace{1.5mu}\lambda^2\!}{\pi^2}\!\int_{\mathbb{R}^6}\mspace{-7.5mu}d\vec{p}\+d\vec{q}\;\frac{4}{\big(\frac{p^2+\,q^2\!}{2}+\frac{\lambda}{n^2}\big)^2\big(p^2+q^2\big)^2}=496\mspace{1.5mu}\pi\lambda^2\!\int_0^{+\infty}\mspace{-24mu}dk\;\frac{k}{\big(\frac{k^2\!}{2}+\frac{\lambda}{n^2}\big)^2\!}\\
            &\mspace{-18mu}=496\mspace{1.5mu}\pi\lambda\mspace{1.5mu}n^2 . 
        \end{align*}
       Using the above estimate, we find 
       \begin{equation*}
            \Phi^\lambda_{\mathrm{off}}[\eta_n]=n^2\Phi^{\mspace{1.5mu}0}_{\mathrm{off}}[f]+ O(n)
        \end{equation*}
        and the lemma is proven.
        
    \end{proof}
\end{lemma}
\noindent In light of lemma~\ref{leadingOrderLemma}, it is straightforward to see that~\eqref{phiUnboundedBelow} is achieved as soon as we exhibit a function $f\!\in\mspace{-1.5mu} H^{1/2}(\mathbb{R}^3)$ such that, whenever $\gamma\!<\mspace{-1.5mu}\gamma_c\+$, holds
\begin{equation*}
    \Phi^{\mspace{1.5mu}0}_{\mathrm{diag}}[f]+\Phi^{\mspace{1.5mu}0}_{\mathrm{off}}[f]+\Phi_{\mathrm{reg}}[f]<0.
\end{equation*}
A relevant feature of the previous lemma is that the leading order of $\Phi^\lambda[\eta_n]$ as $n$ goes to infinity does not depend on $\lambda$ and, therefore, we have reduced the problem to the study of the hermitian quadratic form evaluated in $\lambda=0$ which is diagonalizable.

\noindent In the next lemma we exhibit the trial function we need to prove our result.
\begin{lemma}\label{trialFunctionExhibition}
    Let $\gamma_c$ be defined by~\eqref{criticalGamma}, assume $\gamma\!<\mspace{-1.5mu}\gamma_c$ and let us consider the family of trial functions $f_\beta\!\in\mspace{-1.5mu}H^{1/2}(\mathbb{R}^3)$ such that
    \begin{equation*}
        \hat{f}_\beta(\vec{p})=\tfrac{1}{p^2}\exp\!\left(-\tfrac{p^\beta\mspace{1.5mu}+\,p^{-\beta}}{\mspace{-12mu}2}\mspace{-1.5mu}\right)\!,\qquad\beta>0.
    \end{equation*}
    Then there exists  $ \beta_0>0$ such that for any $\beta \in (0, \beta_0)$ we have 
    \begin{equation*}
        \big(\Phi^{\mspace{1.5mu}0}_{\mathrm{diag}}\!+\Phi^{\mspace{1.5mu}0}_{\mathrm{off}}\mspace{-1.5mu}+\Phi_{\mathrm{reg}}\big)[f_\beta]\mspace{-1.5mu}<0.
    \end{equation*}
    \begin{proof}
        We stress that our trial functions are entirely lying in the $s$-wave subspace, therefore we have
        \begin{subequations}\label{phiSWave}
        \begin{align}
            \Phi^0_{\mathrm{diag}}[f_\beta]&:=48\pi^2\sqrt{\tfrac{3}{4}}\!\int_0^{+\infty}\mspace{-24mu}d\mspace{-1.5mu}p\:p^3\,\lvert\hat{f}_\beta(p)\rvert^2,\label{diagPhiS}\\
            \Phi^0_{\mathrm{off}}[f_\beta]&:=-96\pi\!\int_0^{+\infty}\mspace{-24mu}d\mspace{-1.5mu}p\:p\!\int_0^{+\infty}\mspace{-24mu}dq\:q\,\overline{\hat{f}_\beta(p)}\,\hat{f}_\beta(q)
           \ln\mspace{-1.5mu}\left(\frac{p^2+q^2+pq}{p^2+q^2-pq}\right)\! 
            ,\label{offPhiS}\\
            \Phi_{\mathrm{reg}}[f_\beta]&:=24\pi\gamma\!\!\int_0^{+\infty}\mspace{-24mu}d\mspace{-1.5mu}p\:p\!\int_0^{+\infty}\mspace{-24mu}dq\:q\,\overline{\hat{f}_\beta(p)}\,\hat{f}_\beta(q)
           \ln\mspace{-1.5mu}\left(\frac{p^2+q^2+2pq}{p^2+q^2-2pq}\right)\! 
            ,\label{regPhiS}
        \end{align}
        \end{subequations}
        where we have used the identity\footnote{Equation~\eqref{sWaveAddition} is an application of the addition formula for the spherical harmonics in the $s$-wave.}
        \begin{equation}\label{sWaveAddition}
            \int_{\mathbb{R}^6}\mspace{-7.5mu}d\vec{p}\+d\vec{q}\:g\!\left(p,q,\tfrac{\vec{p}\,\cdot\+\vec{q}}{p\+q}\right)=8\pi^2\!\!\int_0^{+\infty}\mspace{-24mu}d\mspace{-1.5mu}p\:p^2\!\int_0^{+\infty}\mspace{-24mu}dq\:q^2\!\int_{-1}^1\mspace{-7.5mu}du\:g(p,q,u)
        \end{equation}
        holding for any integrable function $g:\mathbb{R}^2_+\times[-1,1]\longrightarrow\mathbb{C}$.
        According to, e.g.~\cite[lemma 3.4]{BCFT}, the quantities in equations~\eqref{phiSWave} can be diagonalized through the unitary transformation
        \begin{equation}\label{MellinTransform}
            \begin{split}
            \mathcal{M}:L^2(\mathbb{R}_+,\,p^2\textstyle{\sqrt{p^2+1}}\,dp)\longrightarrow L^2(\mathbb{R}),\\
            \psi\longmapsto \psi^{\sharp}(x)=\frac{1}{\mspace{-4.5mu}\sqrt{\mspace{1.5mu}2\pi}\,}\!\int_{\mathbb{R}}\mspace{-3mu}dt\:e^{-i\mspace{1.5mu}tx}e^{2t}\psi(e^t)
            \end{split}
        \end{equation}
        yielding (see~\cite[lemmata 3.4, 3.5]{BCFT})
        \begin{subequations}\label{diagonalizedQF}
        \begin{gather}
            \Phi^{\mspace{1.5mu}0}_{\mathrm{diag}}[f_\beta]=48\mspace{1.5mu}\sqrt{\tfrac{3}{4}}\,\pi^2\!
            \int_{\mathbb{R}} \!dx\:\lvert \hat{f}^{\,\sharp}_\beta(x)\rvert^2,\label{diagDiagonalizedQF}\\
            \Phi^{\mspace{1.5mu}0}_{\mathrm{off}}[f_\beta]=-48\pi^2\!\mspace{-1.5mu} \int_{\mathbb{R}}\!dx\:\vert \hat{f}^{\,\sharp}_\beta(x)\rvert^2\,\frac{4\sinh\!\left(\frac{\pi}{6} x\right)}{x\cosh\!\left(\frac{\pi}{2}x\right)},\label{offDiagonalizedQF}\\
            \Phi_{\mathrm{reg}}[f_\beta]=48\pi^2\! \int_{\mathbb{R}}\!dx\: \lvert \hat{f}^{\,\sharp}_\beta(x)\rvert^2\, \frac{\gamma\tanh\!\left(\frac{\pi}{2} x\right)}{x}.\label{regDiagonalizedQF}
        \end{gather}
        \end{subequations}
       Let us introduce the bounded and continuous function
        \begin{equation}
            S(x):=\frac{\sqrt{3}}{2}+\frac{\gamma\sinh\!\left(\frac{\pi}{2}x\right)-4\sinh\!\left(\tfrac{\pi}{6}x\right)}{x\mspace{1.5mu}\cosh\!\left(\frac{\pi}{2}x\right)}
        \end{equation}
        so that we have
        \begin{gather}\label{popo}
            \big(\Phi^{\mspace{1.5mu}0}_{\mathrm{diag}}\!+\Phi^{\mspace{1.5mu}0}_{\mathrm{off}}\mspace{-1.5mu}+\Phi_{\mathrm{reg}}\big)[f_\beta]=48\pi^2\!\!\int_{\mathbb{R}}\!dx\:\lvert\hat{f}^{\,\sharp}_\beta(x)\rvert^2\,S(x),
        \end{gather}
        with
        \begin{equation}
            \lim_{x\rightarrow\, 0} S(x)=\tfrac{\sqrt{3}}{2}-\frac{2\pi}{3}+\frac{\pi}{2}\gamma=\frac{\pi}{2}(\gamma-\gamma_c)<0.
        \end{equation}
        Roughly speaking, the integral in~\eqref{popo} is negative if we choose the trial function such that the support of $\hat{f}^{\,\sharp}_\beta$ is sufficiently concentrated in a neighborhood of zero.
        More precisely, considering the explicit expression of $\hat{f}_\beta\+$, we have\footnote{
        We stress that $\hat{f}^{\,\sharp}_\beta(x)=\sqrt{\tfrac{2}\pi}\,\tfrac{1}{\beta}\,\mathrm{K}_{\mspace{1.5mu}i\mspace{1.5mu}x/\beta}(1)$
        because of the integral representation for the Macdonald function $\mathrm{K}_\nu$ given by~\cite[p. 384, 3.547 4]{GR}.}
        \begin{equation}
            \hat{f}^{\,\sharp}_\beta(x)=\frac{1}{\beta}\,\hat{h}\!\left(\tfrac{x}{\beta}\right)\!
        \end{equation}
        where $h(p) =e^{-\cosh{p}}$ with $h \in \mathcal{S}(\mathbb{R})$. 
        Then
        \begin{align*}
            \int_{\mathbb{R}}\!dx\:\lvert \hat{f}^{\,\sharp}_\beta(x)\rvert^2 S(x)=\frac{1}{\beta^2}\!\int_{\mathbb{R}}\!dx\:\lvert\hat{h}(x/\beta)\rvert^2\,S(x)\\
            =\frac{1}{\beta}\!\int_{\mathbb{R}}\!dx\:\lvert\hat{h}(x)\rvert^2\,S(\beta \mspace{1.5mu}x). 
        \end{align*}
        By dominated convergence theorem we obtain
        \begin{equation*}
            \lim_{\beta\rightarrow\, 0^+}\int_{\mathbb{R}}\!dx\:\lvert\hat{h}(x)\rvert^2\,S(\beta\mspace{1.5mu}x)=\lVert h\rVert^2_{L^2(\mathbb{R})}\,S(0)<0.
        \end{equation*}
        Hence, the lemma is proven by noticing that the previous integral is continuous in $\beta\!>\!0$ and therefore,  the quadratic form $\big(\Phi^{\mspace{1.5mu}0}_{\mathrm{diag}}\!+\Phi^{\mspace{1.5mu}0}_{\mathrm{off}}\mspace{-1.5mu}+\Phi_{\mathrm{reg}}\big)[f_\beta]$ is negative for any $\beta$ small enough.
        
    \end{proof}
\end{lemma}
\begin{proof}[Proof of theorem~\ref{unboundedHamiltonian}]
    Let $\hat{f}_\beta$ be the trial function given in lemma~\ref{trialFunctionExhibition} with $\beta<\beta_0$ and consider the following sequence of charges
    \begin{equation*}
        \hat{\eta}^{\,\beta}_n(\vec{p})=\frac{1}{n}\mspace{1.5mu}\hat{f}_\beta\!\left(\frac{\vec{p}}{n}\right). 
    \end{equation*}
    By lemma~\ref{leadingOrderLemma}, we know that
    $$\Phi^\lambda[\eta^{\,\beta}_n]=n^2(\Phi^0_{\mathrm{diag}}\!+\Phi^0_{\mathrm{off}}\mspace{-1.5mu}+\Phi_{\mathrm{reg}})[f_\beta]+\mathcal{O}(n),\qquad n\rightarrow +\infty$$
    and then $\Phi^\lambda[\eta^{\,\beta}_n]\longrightarrow -\infty$ as $n$ grows to infinity.
    
\end{proof}

\vs\vs
\section{Analysis in Position Space}\label{anapos}
\vs

In this section, we give a different proof of the coercivity of $\Phi^\lambda$ based on the representation of $\Phi^{\lambda}$ in position space.
In particular, this approach allows to identify the negative contribution of the quadratic form $\Phi^{\lambda}_{\mathrm{off}}$ and therefore to justify the choice of the regularization term $\Phi_{\mathrm{reg}}\+$. 

\noindent In the next proposition, we write the quadratic form $\Phi^{\lambda}$ defined in~\eqref{phiDef} in the position-space representation.
\begin{prop}\label{PhiPositionProp}
    For any $\xi\in H^{1/2}(\mathbb{R}^3)$ and $\lambda> 0$ one has
    \begin{subequations}\label{3bodyPhiPosition}
    \begin{gather}
        \Phi^\lambda_{\mathrm{diag}}[\xi]=12\pi\mspace{1.5mu}\sqrt{\lambda}\,\lVert\xi\rVert^2
        +\frac{2\sqrt{3}\mspace{1.5mu}\lambda}{\pi}\!\int_{\mathbb{R}^6}\mspace{-7.5mu}d\vec{x}d\vec{y}\;\frac{\lvert\xi(\vec{x})-\xi(\vec{y})\rvert^2\mspace{-4.5mu}}{\lvert\vec{x}-\vec{y}\rvert^2}\,\mathrm{K}_2\!\left(\sqrt{\tfrac{4\lambda}{3}}\mspace{1.5mu}\lvert\vec{x}-\vec{y}\rvert\right)\!,\label{diagPhiPosition}\\
        \Phi^\lambda_{\mathrm{off}}[\xi]=-\frac{8\sqrt{3}\mspace{1.5mu}\lambda}{\pi}\!\int_{\mathbb{R}^6}\mspace{-7.5mu}d\vec{x}d\vec{y}\;\frac{\overline{\xi(\vec{x})}\,\xi(\vec{y})}{y^2+x^2+\vec{x}\mspace{-1.5mu}\cdot\mspace{-1.5mu}\vec{y}}\,\mathrm{K}_2\!\left(\sqrt{\tfrac{4\lambda}{3}}\,\textstyle{\sqrt{y^2+x^2+\vec{x}\mspace{-1.5mu}\cdot\mspace{-1.5mu}\vec{y}}}\right)\!,\label{offPhiPosition}\\
        \Phi_{\mathrm{reg}}[\xi]=12\pi\mspace{1.5mu}\gamma\!\int_{\mathbb{R}^3}\mspace{-7.5mu}d\vec{x}\;\frac{\lvert\xi(\vec{x})\rvert^2\mspace{-4.5mu}}{\lvert\vec{x}\rvert}\label{regPhiPosition}
    \end{gather}
    \end{subequations}
    where $\mathrm{K}_\nu:\mathbb{R}_+\longrightarrow\mathbb{C}$ is the modified Bessel function of the second kind (also known as Macdonald function) and order $\nu\in\mathbb{C}$.
\begin{proof}
    Identity~\eqref{diagPhiPosition} is a consequence of~\eqref{diagPhi} and~\cite[section 7.12, (5)]{LL}, while~\eqref{regPhiPosition} is obtained by comparing~\eqref{regPhi} with the  identity
    \begin{equation}\label{distributionalFourier}
        \int_{\mathbb{R}^3}\mspace{-7.5mu}d\vec{r}\;\frac{\lvert f(\vec{r})\rvert^2\!}{r}=\frac{1}{2\pi^2\!}\int_{\mathbb{R}^6}\mspace{-7.5mu}d\vec{p}\+d\vec{q}\;\frac{\overline{\hat{f}(\vec{p})}\,\hat{f}(\vec{q})}{\lvert\vec{p}-\vec{q}\rvert^2},\qquad\forall\, f\in H^{1/2}(\mathbb{R}^3).
    \end{equation}
    Concerning the proof of~\eqref{offPhiPosition}, we  consider~\eqref{offPhi} for $\varphi\-\in\-\mathcal{S}(\mathbb{R}^3)$ and we observe that we have uniformly in $\lambda> 0$ 
    \begin{equation*}
        \vec{\sigma}\longmapsto\frac{1}{\sigma^2+\tau^2+\vec{\tau}\!\cdot\mspace{-1.5mu}\vec{\sigma}+\lambda}\in L^2(\mathbb{R}^3,d\vec{\sigma}),\qquad \text{for  } \tau \neq 0.
    \end{equation*}
    Therefore, by Plancherel's theorem we find
    \begin{equation*}
        \Phi^\lambda_{\mathrm{off}}[\varphi]=-\frac{12}{\pi}\!\int_{\mathbb{R}^3}\mspace{-7.5mu}d\vec{\tau}\:\overline{\hat{\varphi}(\vec{\tau})}\!\int_{\mathbb{R}^3}\mspace{-7.5mu}d\vec{x}\;\frac{\varphi(\vec{x})}{(2\pi)^{3/2}\!}\int_{\mathbb{R}^3}\mspace{-7.5mu}d\vec{\sigma}\;\frac{e^{-i\,\vec{x}\cdot\mspace{1.5mu}\vec{\sigma}}}{\tau^2+\sigma^2+\vec{\tau}\!\cdot\mspace{-1.5mu}\vec{\sigma}+\lambda}.
    \end{equation*}
    Using the change of coordinates  $\vec{\sigma}=\vec{q}-\frac{\vec{\tau}}{2}$, we obtain
    \begin{align*}
        \Phi^\lambda_{\mathrm{off}}[\varphi]&=-\frac{12}{\pi}\!\int_{\mathbb{R}^3}\mspace{-7.5mu}d\vec{\tau}\:\overline{\hat{\varphi}(\vec{\tau})}\!\int_{\mathbb{R}^3}\mspace{-7.5mu}d\vec{x}\;\frac{\:e^{\frac{\vec{\tau}\mspace{1.5mu}\cdot\,\vec{x}}{2}i}\!}{(2\pi)^{3/2}\!}\:\varphi(\vec{x})\!\int_{\mathbb{R}^3}\mspace{-7.5mu}d\vec{q}\;\frac{e^{-i\,\vec{q}\cdot\vec{x}}}{\frac{3}{4}\tau^2+q^2+\lambda}\\
        &=-\frac{24\pi}{\:(2\pi)^{3/2}\!}\int_{\mathbb{R}^3}\mspace{-7.5mu}d\vec{\tau}\:\overline{\hat{\varphi}(\vec{\tau})}\!\int_{\mathbb{R}^3}\mspace{-7.5mu}d\vec{x}\;\frac{\varphi(\vec{x})}{\lvert\vec{x}\rvert}\,e^{\frac{\vec{\tau}\mspace{1.5mu}\cdot\,\vec{x}}{2}i\,-\sqrt{\frac{3}{4}\tau^2\,+\,\lambda}\:\lvert\vec{x}\rvert} \\
        &=-\frac{24\pi}{\:(2\pi)^{3/2}\!}\int_{\mathbb{R}^3}\mspace{-7.5mu}d\vec{x}\;\frac{\varphi(\vec{x})}{\lvert\vec{x}\rvert}\!\int_{\mathbb{R}^3}\mspace{-7.5mu}d\vec{\tau}\:\overline{\hat{\varphi}(\vec{\tau})}\,e^{\frac{\vec{\tau}\mspace{1.5mu}\cdot\,\vec{x}}{2}i\,-\sqrt{\frac{3}{4}\tau^2\,+\,\lambda}\:\lvert\vec{x}\rvert}.
    \end{align*}
    Since uniformly in $\lambda> 0$
    \begin{equation*}
        \vec{\tau}\longmapsto e^{\frac{\vec{\tau}\mspace{1.5mu}\cdot\,\vec{x}}{2}\,i\,-\sqrt{\frac{3}{4}\tau^2\,+\,\lambda}\:\lvert\vec{x}\rvert}\in L^2(\mathbb{R}^3,d\vec{\tau}),\qquad \text{for } x \neq 0
    \end{equation*}
    we use again Plancherel's theorem to obtain
    \begin{align*}
        \Phi^\lambda_{\mathrm{off}}[\varphi]&=-\frac{24\pi}{(2\pi)^3}\!\int_{\mathbb{R}^3}\mspace{-7.5mu}d\vec{x}\;\frac{\varphi(\vec{x})}{\lvert\vec{x}\rvert}\!\int_{\mathbb{R}^3}\mspace{-7.5mu}d\vec{y}\:\overline{\varphi(\vec{y})}\!\int_{\mathbb{R}^3}\mspace{-7.5mu}d\vec{\tau}\:e^{i\mspace{1.5mu}\vec{\tau}\cdot\left(\vec{y}\mspace{1.5mu}+\mspace{1.5mu}\frac{\vec{x}}{2}\!\right)\,-\sqrt{\frac{3}{4}\tau^2\,+\,\lambda}\:\lvert\vec{x}\rvert} \\
        &=-\frac{12}{\pi}\!\int_{\mathbb{R}^3}\mspace{-7.5mu}d\vec{x}\;\frac{\varphi(\vec{x})}{\lvert\vec{x}\rvert}\!\int_{\mathbb{R}^3}\mspace{-7.5mu}d\vec{y}\;\frac{\overline{\varphi(\vec{y})}}{\lvert\vec{y}+\frac{\vec{x}}{2}\rvert}\int_0^{+\infty}\mspace{-24mu}d\tau\:\tau\sin\!\left(\tau\lvert\vec{y}+\tfrac{\vec{x}}{2}\rvert\right)e^{-\sqrt{\frac{3}{4}\tau^2\,+\,\lambda}\:\lvert\vec{x}\rvert}.
    \end{align*}
    The last integral can be explicitly computed using the formula (see, e.g., \cite[p. 491, 3.914.6]{GR})
    \begin{equation}
        \int_0^{+\infty}\mspace{-24mu} dx\:x\mspace{1.5mu}\sin(b \mspace{1.5mu}x)\,e^{-\beta\mspace{1.5mu}\sqrt{x^2\,+\,\gamma^2}}\!=\frac{b\mspace{1.5mu}\beta\mspace{1.5mu}\gamma^2}{\beta^2\mspace{-1.5mu}+b^2}\,\mathrm{K}_2\left(\gamma\sqrt{\beta^2\mspace{-1.5mu}+b^2}\right)\!,\quad\forall\, b\in\mathbb{R}\mspace{1.5mu}\text{ and }\beta,\gamma>0
    \end{equation}
    and therefore      identity~\eqref{offPhiPosition} is proven for $\varphi\!\in\mspace{-1.5mu}\mathcal{S}(\mathbb{R}^3)$.
    By a density argument\footnote{Because of propositions~\ref{negativeQFContribution} and~\ref{singularBehaviorPositionRep}, one can obtain a control in the $H^{1/2}$ norm.} the result is extended to any $\varphi\mspace{-1.5mu}\in\! H^{1/2}(\mathbb{R}^3)$.
    
\end{proof}
\end{prop}
\noindent Before proceeding, let us briefly recall some elementary properties of $\mathrm{K}_2(\cdot)$:
\begin{subequations}
\begin{gather}
    x^2\mspace{1.5mu}\mathrm{K}_2(x) \text{ is decreasing in }x>0,\label{decreasingMacDonald}\\
    \mathrm{K}_2(x)=\sqrt{\frac{\pi}{2}}\,e^{-x}\left[\frac{1}{\!\sqrt{x}\:}+\mathcal{O}\!\left(\frac{1}{x^{3/2}\!}\right)\mspace{-1.5mu}\right]\!,\qquad \text{as }x\rightarrow +\infty,\\
    \mathrm{K}_2(x)=\frac{2}{\,x^2\!}-\frac{1}{2}+\mathcal{O}\!\left(x^2\ln{x}\right)\!,\qquad \text{as }x\rightarrow 0^+\!.\label{originMacDonald}
\end{gather}
In particular, notice that~\eqref{decreasingMacDonald} and~\eqref{originMacDonald} imply
\begin{equation}
    \mathrm{K}_2(x)\leq \frac{2}{\,x^2\!}\,,\qquad \forall\, x>0.\label{upperBoundMacDonald}
\end{equation}
\end{subequations}
In the next proposition we show the relevant fact that the negative contribution of $\Phi^\lambda_{\mathrm{off}}$ can be explicitly characterized.
\begin{prop}\label{negativeQFContribution}
    For any  $\varphi\in H^{1/2}(\mathbb{R}^3)$ and $\lambda> 0$ one has
    \begin{equation}\label{singularBehaviorExtracted}
        \begin{split}
        \Phi^\lambda_{\mathrm{off}}[\varphi]=&\,-24\pi\!\int_{\mathbb{R}^3}\mspace{-7.5mu}d\vec{x}\;\frac{\;e^{-\sqrt{\lambda}\,\lvert\vec{x}\rvert}\mspace{-7.5mu}}{\lvert\vec{x}\rvert}\:\lvert\varphi(\vec{x})\rvert^2+\\
        &+\frac{4\sqrt{3}\mspace{1.5mu}\lambda}{\pi}\!\int_{\mathbb{R}^6}\mspace{-7.5mu}d\vec{x}d\vec{y}\;\frac{\lvert\varphi(\vec{x})-\varphi(\vec{y})\rvert^2}{y^2+x^2+\vec{x}\mspace{-1.5mu}\cdot\mspace{-1.5mu}\vec{y}}\,\mathrm{K}_2\!\left(\sqrt{\tfrac{4\lambda}{3}}\,\textstyle{\sqrt{y^2+x^2+\vec{x}\mspace{-1.5mu}\cdot\mspace{-1.5mu}\vec{y}}}\right)\!.
        \end{split}
    \end{equation}
    \begin{proof}
        Let us decompose the expression given by~\eqref{offPhiPosition} as follows
        \begin{equation*}
        \begin{split}
            \Phi^\lambda_{\mathrm{off}}[\varphi]=-\frac{8\sqrt{3}\mspace{1.5mu}\lambda}{\pi}\,\Bigg[\!\int_{\mathbb{R}^3}\mspace{-7.5mu}d\vec{y}\:\lvert\varphi(\vec{y})\rvert^2\!\int_{\mathbb{R}^3}\mspace{-7.5mu}d\vec{x}\;\frac{\mathrm{K}_2\Big(\sqrt{\tfrac{4\lambda}{3}}\,\textstyle{\sqrt{y^2+x^2+\vec{x}\!\cdot\!\vec{y}}}\Big)}{y^2+x^2+\vec{x}\!\cdot\!\vec{y}}\,+\\
            +\!\int_{\mathbb{R}^3}\mspace{-7.5mu}d\vec{y}\:\overline{\varphi(\vec{y})}\!\int_{\mathbb{R}^3}\mspace{-7.5mu}d\vec{x}\;\frac{\varphi(\vec{x})-\varphi(\vec{y})}{y^2+x^2+\vec{x}\!\cdot\!\vec{y}}\,\mathrm{K}_2\!\left(\sqrt{\tfrac{4\lambda}{3}}\,\textstyle{\sqrt{y^2+x^2+\vec{x}\!\cdot\!\vec{y}}}\right)\!\Bigg],
        \end{split}
        \end{equation*}
        Then, we evaluate the first term in the right hand side.
        In proposition~\ref{PhiPositionProp} we have seen that the function
        \begin{equation}
        \begin{split}
            &\hat{f}^\lambda_x:\mathbb{R}^3\longrightarrow\mathbb{R}_+,\qquad x,\lambda\in\mathbb{R}_+,\\
            &\mspace{27mu}\vec{\tau}\longmapsto \frac{e^{-x\mspace{1.5mu}\sqrt{\frac{3}{4}\tau^2\,+\,\lambda}}}{x}
        \end{split}
        \end{equation}
        is such that
        \begin{equation}
            f^\lambda_{\lvert\vec{x}\rvert}\!\left(\vec{y}+\frac{\vec{x}}{2}\right)=\sqrt{\frac{8}{3\pi}}\,\lambda\,\frac{\mathrm{K}_2\Big(\sqrt{\frac{4\lambda}{3}}\,\textstyle{\sqrt{y^2+x^2+\vec{x}\!\cdot\!\vec{y}}}\,\Big)}{y^2+x^2+\vec{x}\!\cdot\!\vec{y}}.
        \end{equation}
        Notice the symmetry in the exchange $\vec{x}\longleftrightarrow\vec{y}$.
        Then,
        \begin{equation*}
            \int_{\mathbb{R}^3}\mspace{-7.5mu}d\vec{x}\:f^\lambda_{\lvert\vec{x}\rvert}\!\left(\vec{y}+\tfrac{\vec{x}}{2}\right)=\!\int_{\mathbb{R}^3}\mspace{-7.5mu}d\vec{x}\:f^\lambda_{\lvert\vec{y}\rvert}\!\left(\vec{x}+\tfrac{\vec{y}}{2}\right)=\!\int_{\mathbb{R}^3}\mspace{-7.5mu}d\vec{z}\:f^\lambda_{\lvert\vec{y}\rvert}(\vec{z})=(2\pi)^{3/2}\hat{f}^\lambda_{\lvert\vec{y}\rvert}(0).
        \end{equation*}
        Therefore we find
        \begin{equation}\label{OldMcDonaldHadYukawa}
            \frac{\lambda}{\!\sqrt{3}\pi^2}\!\int_{\mathbb{R}^3}\mspace{-7.5mu}d\vec{x}\;\frac{\mathrm{K}_2\Big(\sqrt{\tfrac{4\lambda}{3}}\,\textstyle{\sqrt{y^2+x^2+\vec{x}\!\cdot\!\vec{y}}}\,\Big)}{y^2+x^2+\vec{x}\!\cdot\!\vec{y}}=\frac{\:e^{-\sqrt{\lambda}\mspace{1.5mu}\lvert\vec{y}\rvert}\!\!\!}{\lvert\vec{y}\rvert},\qquad\forall\, \lambda> 0,\,\vec{y}\in\mathbb{R}^3\smallsetminus\{\vec{0}\}.
        \end{equation}
        According to~\eqref{OldMcDonaldHadYukawa}, we obtain 
        \begin{equation}\label{singularBehaviorAlmostExtracted}
        \begin{split}
            \Phi^\lambda_{\mathrm{off}}[\varphi]=&-24\pi\!\int_{\mathbb{R}^3}\mspace{-7.5mu}d\vec{y}\:\lvert\varphi(\vec{y})\rvert^2\,\frac{\:e^{-\sqrt{\lambda}\mspace{1.5mu}\lvert\vec{y}\rvert}\!\!\!}{\lvert\vec{y}\rvert}\;+\\
            &+\frac{8\sqrt{3}\mspace{1.5mu}\lambda}{\pi}\!\int_{\mathbb{R}^6}\mspace{-7.5mu}d\vec{x}d\vec{y}\;\frac{\overline{\varphi(\vec{y})}\,[\varphi(\vec{y})-\varphi(\vec{x})]}{y^2+x^2+\vec{x}\!\cdot\!\vec{y}}\,\mathrm{K}_2\Big(\sqrt{\tfrac{4\lambda}{3}}\,\textstyle{\sqrt{y^2+x^2+\vec{x}\!\cdot\!\vec{y}}}\,\Big).
        \end{split}
        \end{equation}
        It is now sufficient to notice that the symmetry in exchanging $\vec{x}\longleftrightarrow\vec{y}$ allows us to write
        \begin{align*}
            \int_{\mathbb{R}^3}\mspace{-7.5mu}d\vec{y}\:\overline{\varphi(\vec{y})}\!\int_{\mathbb{R}^3}\mspace{-7.5mu}d\vec{x}\;&\frac{\varphi(\vec{y})-\varphi(\vec{x})}{y^2+x^2+\vec{x}\!\cdot\!\vec{y}}\,\mathrm{K}_2\Big(\sqrt{\tfrac{4\lambda}{3}}\,\textstyle{\sqrt{y^2+x^2+\vec{x}\!\cdot\!\vec{y}}}\,\Big)=\\
            &=\frac{1}{2}\int_{\mathbb{R}^6}\mspace{-7.5mu}d\vec{x}d\vec{y}\;\frac{\lvert\varphi(\vec{y})-\varphi(\vec{x})\rvert^2\!}{y^2+x^2+\vec{x}\!\cdot\!\vec{y}}\,\mathrm{K}_2\Big(\sqrt{\tfrac{4\lambda}{3}}\,\textstyle{\sqrt{y^2+x^2+\vec{x}\!\cdot\!\vec{y}}}\,\Big).
        \end{align*}
        and the proposition is proved.
        
    \end{proof}
\end{prop}
\noindent Thanks to proposition~\ref{negativeQFContribution}, it is not hard to find lower and upper bounds for $\Phi^\lambda$.
To this end, it is convenient to introduce the Gagliardo semi-norm of the Sobolev space $H^{1/2}(\mathbb{R}^3)$, defined as
\begin{equation}\label{GagliardoHilbertSobolev}
    [\mspace{1.5mu}u\mspace{1.5mu}]^2_{\frac{1}{2}}\!:=\!\int_{\mathbb{R}^6}\mspace{-7.5mu}d\vec{x}d\vec{y}\;\frac{\lvert u(\vec{x})\--u(\vec{y})\rvert^2\!\!}{\;\lvert\vec{x}-\vec{y}\rvert^4},\qquad u\in H^{1/2}(\mathbb{R}^3),
\end{equation}
so that $\lVert u\rVert_{H^{1/2}(\mathbb{R}^{3})}^2\!=\lVert u\rVert^2\!+[\mspace{1.5mu}u\mspace{1.5mu}]^2_{\frac{1}{2}}$.
In terms of the Fourier transform of $u$ we also have (see e.g.,~\cite[proposition 3.4]{DNPV})
\begin{equation}\label{GagliardoHilbertSobolevFourier}
   [\mspace{1.5mu}u\mspace{1.5mu}]^2_{\frac{1}{2}}\!=2\pi^2\mspace{-4.5mu}\int_{\mathbb{R}^3}\mspace{-7.5mu}d\vec{k}\:\lvert\vec{k}\rvert\mspace{1.5mu}\lvert\hat{u}(\vec{k})\rvert^2.
\end{equation}
\begin{prop}\label{singularBehaviorPositionRep} For any given $\varphi\in H^{1/2}(\mathbb{R}^3)$, one has
\begin{subequations}\label{ThetaBoundsPosition}
\begin{gather}
    \Phi^\lambda[\varphi]\geq\Phi^\lambda_{\mathrm{diag}}[\varphi]+3\min\{0,\gamma-2\}[\mspace{1.5mu}\varphi\mspace{1.5mu}]^2_{\frac{1}{2}},\label{ThetaLowerBoundPosition}\\
    \Phi^\lambda[\varphi]\leq\Phi^\lambda_{\mathrm{diag}}[\varphi]+\left(3\gamma+\tfrac{96\sqrt{3}}{\pi}\right)\![\mspace{1.5mu}\varphi\mspace{1.5mu}]^2_{\frac{1}{2}} \label{ThetaUpperBoundPosition}.
\end{gather}
\end{subequations}
\begin{proof}
    The lower bound is obtained by neglecting the positive part of $\Phi^\lambda_{\mathrm{off}}$
    \begin{align*}
        \Phi^\lambda_{\mathrm{off}}[\varphi]+\Phi_{\mathrm{reg}}[\varphi]\geq 12\pi\!\int_{\mathbb{R}^3}\mspace{-7.5mu}d\vec{y}\;\frac{\lvert\varphi(\vec{y})\rvert^2\!}{\lvert\vec{y}\rvert}\left(\gamma-2\mspace{1.5mu}e^{-\sqrt{\lambda}\,\lvert\vec{y}\rvert}\right)
    \end{align*}
    and by considering the following inequalities
    \begin{gather}
        \inf_{\vec{y}\in\,\mathbb{R}^3}\!\left\{\gamma-2\mspace{1.5mu}e^{-\sqrt{\lambda}\,\lvert\vec{y}\rvert}\right\}\!\geq \min\{0,\gamma-2\},\\[-2.5pt]
        \int_{\mathbb{R}^3}\mspace{-7.5mu}d\vec{x}\;\frac{\lvert\varphi(\vec{x})\rvert^2\!}{\lvert\vec{x}\rvert}\leq \frac{1}{4\pi}[\mspace{1.5mu}\varphi\mspace{1.5mu}]^2_{\frac{1}{2}}.\label{singularityIsGagliardo}
    \end{gather}
    Notice that~\eqref{singularityIsGagliardo} is a consequence of the Hardy-Rellich inequality (see\footnote{There is a typo in~\cite[equation~(1.4)]{Y}: a power $2$ on the Euler Gamma function in the numerator is missing.}~\cite{Y})
    \begin{equation}
        \int_{\mathbb{R}^3}\mspace{-7.5mu}d\vec{x}\;\frac{\lvert u(\vec{x})\rvert^2\!}{\lvert\vec{x}\rvert}\leq\frac{\pi}{2} \int_{\mathbb{R}^3}\mspace{-7.5mu}d\vec{k}\:\lvert\vec{k}\rvert\mspace{1.5mu}\lvert\hat{u}(\vec{k})\rvert^2,\qquad \forall\,u\in H^{1/2}(\mathbb{R}^3)
    \end{equation}
    compared with~\eqref{GagliardoHilbertSobolevFourier}.
    In order to obtain the upper bound, we recall~\eqref{GagliardoHilbertSobolev} to get
    \begin{align*}
        \int_{\mathbb{R}^6}\mspace{-7.5mu}d\vec{x}d\vec{y}\;\frac{\lvert\varphi(\vec{y})-\varphi(\vec{x})\rvert^2\!}{y^2+x^2+\vec{x}\!\cdot\!\vec{y}}&\,\mathrm{K}_2\Big(\sqrt{\tfrac{4\lambda}{3}}\,\textstyle{\sqrt{y^2+x^2+\vec{x}\!\cdot\!\vec{y}}}\,\Big)\!\leq\\
        &\leq[\mspace{1.5mu}\varphi\mspace{1.5mu}]_{\frac{1}{2}}^2\sup_{(\vec{x},\,\vec{y})\in\,\mathbb{R}^6}\,\frac{\lvert\vec{x}-\vec{y}\rvert^4}{y^2+x^2+\vec{x}\!\cdot\!\vec{y}}\,\mathrm{K}_2\Big(\sqrt{\tfrac{4\lambda}{3}}\,\textstyle{\sqrt{y^2+x^2+\vec{x}\!\cdot\!\vec{y}}}\,\Big).
    \end{align*}
    We make use of~\eqref{upperBoundMacDonald} and get rid of the dependence on the angles in evaluating the $\sup$, since
    \begin{equation*}
    \begin{cases}
        x^2+y^2-2\,\vec{x}\!\cdot\!\vec{y}\leq 2\mspace{1.5mu}(x^2+y^2),\\
        x^2+y^2+\vec{x}\!\cdot\!\vec{y}\geq \frac{x^2+y^2}{2},
        \end{cases}\mspace{-18mu}\implies \frac{\lvert\vec{x}-\vec{y}\rvert^2}{x^2+y^2+\vec{x}\!\cdot\!\vec{y}}\leq 4.  
    \end{equation*}
    Hence,
    \begin{align*}
        \int_{\mathbb{R}^6}\mspace{-7.5mu}d\vec{x}d\vec{y}\;\frac{\lvert\varphi(\vec{y})-\varphi(\vec{x})\rvert^2\!\-}{y^2+x^2+\vec{x}\!\cdot\!\vec{y}}\,&\mathrm{K}_2\Big(\sqrt{\tfrac{4\lambda}{3}}\,\textstyle{\sqrt{y^2+x^2+\vec{x}\!\cdot\!\vec{y}}}\Big)\leq\\
        &\leq\frac{3}{2\mspace{1.5mu}\lambda}\,[\mspace{1.5mu}\varphi\mspace{1.5mu}]_{\frac{1}{2}}^2\sup_{(\vec{x},\,\vec{y})\in\,\mathbb{R}^6}\frac{\lvert\vec{x}-\vec{y}\rvert^4}{\!\left(y^2+x^2+\vec{x}\!\cdot\!\vec{y}\right)^{\!2}\!}=\frac{24}{\lambda}\,[\mspace{1.5mu}\varphi\mspace{1.5mu}]_{\frac{1}{2}}^2.
    \end{align*}
    We stress that this estimate is optimal, since the argument of the supremum in $\mathbb{R}^6$ attains the previous upper bound along the hyperplane $\vec{x}+\vec{y}=\vec{0}$.
    
    \noindent So far, we have obtained for any $\varphi\in H^{1/2}(\mathbb{R}^3)$
    \begin{equation}\label{nonSingularFinalEstimate}
        \Phi^\lambda_{\mathrm{off}}[\varphi]\leq - 24\pi\!\int_{\mathbb{R}^3}\mspace{-7.5mu}d\vec{y}\;\frac{\lvert\varphi(\vec{y})\rvert^2\!}{\lvert\vec{y}\rvert} \,e^{-\sqrt{\lambda}\,\lvert\vec{y}\rvert}+
        \frac{96\sqrt{3}}{\pi}\,[\mspace{1.5mu}\varphi\mspace{1.5mu}]_{\frac{1}{2}}^2.
    \end{equation}
    We complete the proof simply by neglecting the negative contribution.
    
    \end{proof}
\end{prop}
\noindent The major difficulties in the proof of the coercivity of $\Phi^\lambda$ in momentum space  obtained in~\cite{BCFT} lie in the search of a lower bound.
On the other hand, in position space such estimate, provided in  proposition~\ref{singularBehaviorPositionRep}, turns out to be much easier. 
However, some accuracy is lost in this framework.
Indeed, adopting~\eqref{ThetaLowerBoundPosition} to obtain an estimate from below for $\Phi^\lambda$, one gets
\begin{equation}
    \Phi^\lambda[\xi] \geq\Phi^\lambda_{\mathrm{diag}}[\xi]-3\max\{0,2-\gamma\}[\mspace{1.5mu}\xi\mspace{1.5mu}]^2_{\frac{1}{2}}+12\pi\,\mathop{\mathrm{essinf}}_{\vec{y}\,\in\,\mathbb{R}^3}\beta(\vec{y})\,\lVert \xi\rVert^2_{L^2(\mathbb{R}^3)}
\end{equation}
or, equivalently
\begin{equation}\label{toBeChosenPositive}
    \Phi^\lambda[\xi]\mspace{-1.5mu} \geq\!\int_{\mathbb{R}^3}\mspace{-7.5mu}d\vec{p}\left[12\pi\sqrt{\tfrac{3}{4}p^2+\lambda}-6\pi^2\max\{0,2\mspace{-1.5mu}-\mspace{-1.5mu}\gamma\}p+12\pi\,\mathop{\mathrm{essinf}}_{\vec{y}\,\in\,\mathbb{R}^3}\beta(\vec{y})\right]\mspace{-1.5mu}\lvert\xi(\vec{p})\rvert^2.
\end{equation}
The function in square brackets attains its minimum at
\begin{equation*}
    p_{\mathrm{min}}=\frac{\pi\max\{0,2\mspace{-1.5mu}-\mspace{-1.5mu}\gamma\}\sqrt{\lambda}}{\!\sqrt{9-3\pi^2\max\{0,2\mspace{-1.5mu}-\mspace{-1.5mu}\gamma\}^2}\,},
\end{equation*}
provided
\begin{equation}
    3-\pi^2\max\{0,2-\gamma\}^2> 0\iff \gamma> 2-\frac{\sqrt{3}}{\pi}=:\gamma_c'.
\end{equation}
Indeed, plugging the value of $p_{\mathrm{min}}$ in the right hand side of~\eqref{toBeChosenPositive}, one gets
\begin{equation}
    \Phi^\lambda[\xi]\mspace{-1.5mu} \geq \left(4\sqrt{3}\pi\mspace{1.5mu}\sqrt{\lambda}\sqrt{3-\pi^2\max\{0,2\mspace{-1.5mu}-\mspace{-1.5mu}\gamma\}^2}+12\pi\mathop{\mathrm{essinf}}_{\vec{y}\,\in\,\mathbb{R}^3}\beta(\vec{y})\right)\lVert \xi\rVert^2_{L^2(\mathbb{R}^3)}
\end{equation}
that is positive for
\begin{equation}
    \lambda>\frac{3\min\{0,\mathop{\mathrm{essinf}} \beta\}^2}{3-\pi^2\max\{0,2\mspace{-1.5mu}-\mspace{-1.5mu}\gamma\}^2}.
\end{equation}
As mentioned above, $\gamma'_c\approx 1.44867$ is not optimal, since $\gamma'_c\gneq\gamma_c\+$.
\vs\vs

\vs\vs
\end{document}